\documentclass[11pt]{article}

\usepackage{times}
\usepackage{fullpage,amsmath,amsthm}
\usepackage{amssymb,latexsym,graphicx}

\usepackage{url}

\newcommand{\sgn}{\mathrm{sgn}}
\newcommand{\dist}{\Delta}

\def\01{\{0,1\}}

\newcommand{\R}{\ensuremath{\mathbb{R}}}

\newcommand{\Sphere}{\ensuremath{\mathbb{S}}}

\newcommand{\ang}[1]{\langle #1 \rangle}
\newtheorem{theorem}{Theorem}

\newtheorem{proposition}[theorem]{Proposition}
\newtheorem{lemma}[theorem]{Lemma}

\newtheorem{claim}[theorem]{Claim}
\newtheorem{fact}[theorem]{Fact}
\newtheorem{corollary}[theorem]{Corollary}

\newcommand{\eps}{\varepsilon}

\newcommand{\ignore}[1]{}

\newcommand{\gap}{\gamma}
\newcommand{\GHD}[2]{\textsc{GHD}_{#1,#2}}
\newcommand{\GHS}[1]{\ensuremath{\mathcal{GHD}_{#1}}}

\newcommand{\deq}{:=}

\begin{document}

\title{Better Gap-Hamming Lower Bounds via Better Round Elimination}
\author{Joshua Brody\thanks{Department of Computer Science, Dartmouth College, Hanover, NH 03755. Supported in part by NSF Grant CCF-0448277. Part of this work was
done while the author was visiting CWI and Tel Aviv University.}
\and Amit Chakrabarti\thanks{Department of Computer Science, Dartmouth College, Hanover, NH 03755. Supported in part by NSF Grants CCF-0448277 and IIS-0916565 and a McLane
Family Fellowship.}
\and Oded Regev\thanks{Blavatnik School of Computer Science, Tel Aviv University, Tel Aviv 69978, Israel. Supported by the Israel Science Foundation, by the European Commission under the Integrated Project QAP funded by the IST directorate as Contract Number 015848, by the Wolfson Family Charitable Trust, and by a European Research Council (ERC) Starting Grant.}
\and Thomas Vidick\thanks{UC Berkeley, vidick@eecs.berkeley.edu. Supported by ARO Grant W911NF-09-1-0440 and NSF Grant CCF-0905626. Part of this work was done while the author was visiting CWI and Tel Aviv University.}
\and Ronald de Wolf\thanks{CWI Amsterdam, rdewolf@cwi.nl. Supported by a Vidi grant from the Netherlands Organization for Scientific Research (NWO).}
}
\date{}
\maketitle

\begin{abstract}
Gap Hamming Distance is a well-studied problem in communication
complexity, in which Alice and Bob have to decide whether the Hamming
distance between their respective $n$-bit inputs is less than
$n/2-\sqrt{n}$ or greater than $n/2+\sqrt{n}$.  We show that every
$k$-round bounded-error communication protocol for this problem sends
a message of at least $\Omega(n/(k^2\log k))$ bits.  This lower bound
has an exponentially better dependence on the number of rounds than
the previous best bound, due to Brody and Chakrabarti.  Our
communication lower bound implies strong space lower bounds on algorithms
for a number of data stream computations, such as 
approximating the number of distinct elements in a stream.  

Subsequent to this result, the bound has been improved by some of us to the optimal $\Omega(n)$, independent of $k$, by using different techniques.
\end{abstract}

\section{Introduction}

\subsection{The communication complexity of the Gap-Hamming problem}

Communication complexity studies the communication requirements of
distributed computing.  In its simplest and best-studied setting, two
players, Alice and Bob, receive inputs $x$ and $y$, respectively, and
are required to compute some function $f(x,y)$.  Clearly, for most
functions $f$, the two players need to communicate to solve this
problem.  The basic question of communication complexity is the
\emph{minimal amount} of communication needed.  By abstracting away
from the resources of local computation time and space, communication
complexity gives us a bare-bones but elegant model of distributed
computing.  It is useful and interesting for its own sake, but also
one of our main sources of lower bounds in many other models of
computation, such as data structures, circuit size and depth, Turing
machines, VLSI, and algorithms for data streams.  The basic results
are excellently covered in the book of Kushilevitz and
Nisan~\cite{kushilevitz&nisan}, but many more fundamental results have
appeared since its publication in 1997.

One of the few basic problems whose randomized communication
complexity is not yet well-understood, is the \emph{Gap Hamming
  Distance} (GHD) problem, defined as follows.
\begin{quote}
GHD: Alice receives input $x\in\01^n$ and Bob receives input
$y\in\01^n$, with the promise that $|\Delta(x,y)-n/2| \geq \sqrt{n}$,
where $\Delta$ denotes Hamming distance.  Decide whether $\Delta(x,y)
< n/2$ or $\Delta(x,y) > n/2$.
\end{quote}
Mind the gap between $n/2-\sqrt{n}$ and $n/2+\sqrt{n}$, which is what
makes this problem interesting and useful. We will be concerned with
the communication complexity of randomized protocols that solve GHD. A
gap size of $\Theta(\sqrt{n})$ is the natural choice -- it is where a
$\Theta(1)$ fraction of the inputs lie inside the promise area, and as
we'll see below, it is precisely this choice of gap size that has strong
implications for streaming algorithms lower bounds. Moreover, understanding
 the complexity of the $\sqrt{n}$-gap version can be shown to imply a 
 complete understanding of the GHD problem for all gaps.
  The communication
complexity of the gapless version, where there is no promise on the
inputs, can easily be seen to be linear (for instance by a reduction
from disjointness).  The gap makes the problem easier, and the
question is how it affects the communication complexity: does it
remain linear?

Protocols for GHD and more general problems can be obtained by
sampling.  Suppose for instance that either
$\Delta(x,y)\leq(1/2-\gamma)n$ or $\Delta(x,y)\geq(1/2+\gamma)n$.
Choosing an index $i\in[n]$ at random, the predicate $[x_i\neq y_i]$
is a coin flip with heads probability $\leq 1/2-\gamma$ in the first
case and $\geq 1/2+\gamma$ in the second.  It is known that flipping
such a coin $O(1/\gamma^2)$ times suffices to distinguish these two
cases with probability at least $2/3$.  If we use shared randomness to
choose $O(1/\gamma^2)$ indices, we obtain a one-round bounded-error
protocol with communication $O(1/\gamma^2)$ bits.  In
particular, for GHD (where $\gamma=1/\sqrt{n}$), the communication is
$O(n)$ bits, which is no better than the trivial upper bound of $n$
when Alice just sends $x$ to Bob.

What about lower bounds?  Indyk and Woodruff~\cite{indyk&woodruff}
managed to prove a linear lower bound for the case of one-round
protocols for GHD, where there is only one message from Alice to Bob
(see also~\cite{woodruff,Jayram-Kumar-Sivakumar}).  However, going
beyond one-round bounds turned out to be quite a difficult problem.
Recently, Brody and Chakrabarti~\cite{brody&chakrabarti:gh} obtained
linear lower bounds for all \emph{constant}-round protocols:
\begin{theorem}
\cite{brody&chakrabarti:gh} Every $k$-round bounded-error protocol for
GHD communicates at least $\displaystyle \frac{n}{2^{O(k^2)}}$ bits.
\end{theorem}
In fact their bound is significant as long as the number of rounds is
$k\leq c_0\sqrt{\log n}$, for a universal constant $c_0$.  Regarding
lower bounds that hold irrespective of the number of rounds, an easy
reduction gives an $\Omega(\sqrt{n})$ lower bound (which is folklore):
take an instance of the gapless version of the problem on
$x,y\in\01^{\sqrt{n}}$ and ``repeat'' $x$ and $y$ $\sqrt{n}$ times each.  This
blows up the gap from 1 to $\sqrt{n}$, giving an instance of GHD on
$n$ bits.  Solving this $n$-bit instance of GHD solves the
$\sqrt{n}$-bit instance of the gapless problem.  Since we have a
linear lower bound for the latter, we obtain a general
$\Omega(\sqrt{n})$ bound for GHD.%
\footnote{In fact the same proof lower-bounds the \emph{quantum}
  communication complexity; a linear quantum lower bound for the
  gapless version follows easily from Razborov's work~\cite{razborov}
  and the observation that $\Delta(x,y)=|x|+|y|-2|x\wedge y|$.
  However, as Brody and Chakrabarti observed, in the quantum case this
  $\sqrt{n}$ lower bound is essentially tight: there is a
  bounded-error quantum protocol that communicates $O(\sqrt{n}\log n)$
  qubits.  This also implies that lower bound techniques which apply
  to quantum protocols, such as discrepancy, factorization
  norms~\cite{linial&shraibman,Lee&shraibman}, and the pattern matrix
  method~\cite{sherstov}, cannot prove better bounds for classical
  protocols.}

\subsection{Our results}

Our main result is an improvement of the bound of Brody and Chakrabarti, with an exponentially better dependence
on the number of rounds:
\begin{theorem} \label{thm:main-informal}
Every $k$-round bounded-error protocol for GHD sends a message of length $\displaystyle\Omega\left(\frac{n}{k^2\log k}\right)$.
\end{theorem}

In fact we get a bound for the more general problem of distinguishing distance 
$\Delta(x,y)\leq(1/2-\gamma)n$ from $\Delta(x,y)\geq(1/2+\gamma)n$, as long as $\gamma=\Omega(1/\sqrt{n})$:
for this problem every $k$-round protocol sends a message of $\Omega\left(\frac{1}{k^2\log k}\frac{1}{\gamma^2}\right)$ bits.

Like the result of~\cite{brody&chakrabarti:gh}, our lower bound
deteriorates with the number of rounds.  Also like their result, our
proof is based on \emph{round elimination}, an important
framework for proving communication lower bounds.
Our proof contains an important insight into
this framework that we now explain.

A communication problem usually involves a number of parameters, such as 
the input size, an error bound, and in our case the gap size. 
The round elimination framework consists of showing that a $k$-round
protocol solving a communication problem for a class $\mathcal{C}$
of parameters can be turned into a $(k-1)$-round protocol for an easier class
$\mathcal{C'}$, provided the message communicated in the first round is
short. This fact is then applied repeatedly to obtain a $0$-round
protocol (say), for some nontrivial class of instances. The resulting
contradiction can then be recast as a communication lower bound.
Historically, the easier class $\mathcal{C'}$ has contained {\em smaller
input sizes}\footnote{In fact, the classes
$\mathcal{C}$ and $\mathcal{C'}$ are often designed in such a way that
an instance in $\mathcal{C}$ is a ``direct sum'' of several independent
instances in $\mathcal{C'}$} than those in $\mathcal{C}$. 

In contrast to previous applications of round elimination, we manage
to \emph{avoid shrinking the input length}: the simplification will
instead come from a slight deterioration in the error parameter. Here is how
this works.  If Alice's first message is short, then there is a
specific message and a large set $A$ of inputs on which Alice would
have sent that message.  Roughly speaking, we can use the largeness of
$A$ to show that \emph{almost any} input $\tilde{x}$ for Alice is
close to $A$ in Hamming distance. Therefore, Alice can ``move''
$\tilde{x}$ to its nearest neighbor, $x$, in $A$: this make her first
message redundant, as it is constant for all inputs $x\in A$.  Since
$x$ and $\tilde{x}$ have small Hamming distance, it is likely that
both pairs $(\tilde{x},y)$ and $(x,y)$ are on the same side of the
gap, i.e.  have the same GHD value. Hence the correctness of the new
protocol, which is one round shorter, is only mildly affected by the
move.  Eliminating all $k$ rounds in this manner, while carefully
keeping track of the accumulating errors, yields a lower bound of
$\Omega(n/(k^4\log^2 k))$ on the maximum message length of any
$k$-round bounded-error protocol for GHD.

Notice that this lower bound is slightly weaker than the above-stated
bound of $\Omega(n/(k^2\log k))$. To obtain the stronger bound, we
leave the purely combinatorial setting 
and analyze a version of GHD \emph{on the sphere}:%
\footnote{The idea of
going to the sphere was also used by Jayram et
al.~\cite{Jayram-Kumar-Sivakumar} for a simplified one-round lower
bound. As we will see in Section~\ref{sec:preliminaries}, doing so
is perhaps even more natural than working with the combinatorial version; in particular
it is then easy to make GHD into a \emph{dimension-independent} problem.}
  Alice's input is a unit vector $x\in\mathbb{R}^n$ and Bob's
input is a unit vector $y\in\mathbb{R}^n$, with the promise that either
$x\cdot y\geq 1/\sqrt{n}$ or $x\cdot y\leq -1/\sqrt{n}$ (as we show
below in Section~\ref{sec:preliminaries}, this version and the Boolean
one are essentially equivalent in terms of communication complexity).
Alice's input is now close to the large, constant-message set $A$ in
\emph{Euclidean distance}. The rest of the proof is as outlined above,
but the final bound is stronger than in the combinatorial proof for
reasons that are discussed in Section~\ref{sec:measconc}.  Although this
proof uses arguments from high-dimensional geometry, such as measure
concentration, it arguably remains conceptually simpler than the one
in~\cite{brody&chakrabarti:gh}. 


\paragraph{Related work.} The round elimination technique was formally
identified and named in Miltersen et al.~\cite{MiltersenNSW98} and dates
back even further, at least to Ajtai's lower bound for predecessor data
structures~\cite{Ajtai88}.
For us, the most relevant previous use of this technique is in the result by
Brody and Chakrabarti~\cite{brody&chakrabarti:gh}, where a
weaker lower-bound is proved on GHD.

Their proof, as ours, identifies a large subset $A$ of inputs on which Alice
sends the same message. The ``largeness" of $A$
is used to identify a suitable subset of $(n/3)$ coordinates such that
Alice can ``lift'' any $(n/3)$-bit input $\tilde{x}$, defined on these
coordinates, to some $n$-bit input $x \in A$. In the resulting
protocol for $(n/3)$-bit inputs, the first message is now constant,
hence redundant, and can be eliminated. 

The input size thus shrinks from $n$ to $n/3$ in one round elimination
step. As a result of this constant-factor shrinkage, the
Brody-Chakrabarti final lower bound necessarily decays exponentially
with the number of rounds.  Our proof crucially avoids this shrinkage
of input size by instead considering the {\em geometry} of the set
$A$, and exploiting the natural invariance of the GHD predicate to
small perturbations of the inputs.

\paragraph{Remark.}
This round elimination result was obtained in July 2009.
Soon after, in August 2009, the bound was actually improved by some us of to the optimal $\Omega(n)$ independent of the number of rounds, see~\cite{paper2}.
However, the techniques used are completely different, and as such we feel our result and its proof are of independent interest.

\subsection{Applications to streaming}

The study of gapped versions of the Hamming distance problem by Indyk
and Woodruff~\cite{indyk&woodruff} was motivated by the streaming
model of computation, in particular the problem of approximating the
number of distinct elements in a data stream.  For many data stream
problems, including the distinct elements problem, the goal is to
output a multiplicative approximation of some real-valued quantity.  Usually, both
\emph{randomization} and \emph{approximation} are required.
When both are allowed, there are often remarkably space-efficient
solutions.  

As Indyk and Woodruff showed, \emph{communication lower
  bounds} for gapped versions of the Hamming distance problem imply
\emph{space lower bounds} on algorithms that output the number of
distinct elements in a data stream up to a multiplicative approximation factor
$1\pm\gamma$.
The reduction from the gapped version of Hamming distance works as
follows.  Alice converts her $n$-bit string $x = x_1x_2\cdots x_n$
into a stream of tuples $\sigma = \ang{(1,x_1), (2,x_2), \ldots,
  (n,x_n)}$.  Bob converts $y$ into $\tau = \ang{(1,y_1), (2,y_2),
  \ldots, (n,y_n)}$ in a similar fashion.  Using a streaming algorithm
for the distinct elements problem, Alice processes $\sigma$ and sends
the memory contents to Bob, who then processes $\tau$ starting from
where Alice left off.  In this way, they estimate the number of
distinct elements in $\sigma \circ \tau$.  Note that each element in
$\sigma$ is unique, and that elements in $\tau$ are distinct from
elements in $\sigma$ precisely when $x_i \neq y_i$.  Hence, an
accurate approximation ($\gamma = \Omega(1/\sqrt{n})$ is required) for the
number of distinct elements in $\sigma \circ \tau$ gives an answer to
GHD.  This reduction can be extended to multi-pass streaming
algorithms in a natural way: when Bob is finished processing $\tau$,
he sends the memory contents back to Alice, who begins processing
$\sigma$ a second time.  Generalizing, it is easy to see that a
$p$-pass streaming algorithm gives a $(2p-1)$-round communication
protocol, where each message is the memory contents of the streaming
algorithm.  Accordingly, a lower bound on the length of the largest
message of $(2p-1)$-round protocols gives a space lower bound for the
$p$-pass streaming algorithm.

Thus, the one-round linear lower bound by Indyk and
Woodruff~\cite{indyk&woodruff} yields the desired $\Omega(1/\gamma^2)$
(one-pass) space lower bound for the streaming problem.  Similarly,
our new communication lower bounds imply
$\displaystyle\Omega(1/(\gamma^2 p^2 \log p))$ space lower bounds for
$p$-pass algorithms for the streaming problem.  This bound is
$\Omega(1/\gamma^{2-o(1)})$ for all $p = n^{o(1)}$ and improves on
previous bounds for all $p = o(n^{1/4}/\sqrt{\log n})$.

\paragraph{Organization of the paper.} We start with some preliminaries
in Section~\ref{sec:preliminaries}, including a discussion of the key
measure concentration results that we will use, both for the sphere and
for the Hamming cube, in Section~\ref{sec:measconc}. In
Section~\ref{sec:main} we prove our main result, while in
Section~\ref{sec:combinatorial} we give the simple combinatorial proof
of the slightly weaker result mentioned above.

\section{Preliminaries}\label{sec:preliminaries}

\paragraph{Notation.} 
For $x,y\in \R^n$, let $d(x,y):=\|x-y\|$ be the Euclidean distance
between $x$ and $y$.  For $z \in \R$, define $\sgn(z) \deq 0$ if $z
\geq 0$, and $\sgn(z) = 1$ otherwise.  For a set $S\subseteq \R^n$,
let $d(x,S)$ be the infimum over all $y\in S$ of $d(x,y)$.  The unique
rotationally-invariant probability distribution on the 
$n$-dimensional sphere $\Sphere^{n-1}$ is the Haar measure, which we
denote by~$\nu$. When we say that a vector is taken from the uniform 
distribution over a measurable subset of the sphere, we will always mean
that it is distributed according to the Haar measure, conditioned on 
being in that subset.

Define the max-cost of a communication protocol to be the length of the
longest {\em single} message sent during an execution of the protocol,
for a worst-case input. We use $R^k_\eps(f)$ to denote the minimal
max-cost amongst all two-party, $k$-round, public-coin
protocols that compute $f$ with error probability at most $\eps$ on every
input (here a ``round'' is one message).  See~\cite{kushilevitz&nisan}
for precise definitions.



\subsection{Problem definition}

We will prove our lower bounds for the problem $\GHS{d,\gap}$, where
$d$ is an integer and $\gap>0$. In this problem Alice receives a
$d$-dimensional unit vector $x$, and Bob receives a $d$-dimensional unit vector $y$,
with the promise that $|x\cdot y|\geq \gap$.  Alice and Bob should
output $\sgn(x\cdot y)$.

We show that $\GHS{n,1/\sqrt{n}}$ has essentially the same randomized
communication complexity as the problem GHD that we defined in the introduction. 
Generalizing that definition, for any
$g>0$ define the problem $\GHD{n}{g}$, in which the input is formed of
two $n$-bit strings $x$ and $y$, with the promise that
$|\Delta(x,y)-n/2|\geq g$, where $\Delta$ is the Hamming
distance. Alice and Bob should output $0$ if $\Delta(x,y) < n/2$ and
$1$ otherwise.

The following proposition shows that for any $\sqrt{n}\leq g \leq n$,
the problems $\GHD{n}{g}$ and $\GHS{d,\gap}$ are essentially equivalent from the point
of view of randomized communication complexity (with shared randomness) as long as $d\geq n$ and $\gap = \Theta(g/n)$. 
This proposition also shows that the randomized communication complexity of $\GHS{d,\gap}$ is independent of 
the dimension $d$ of the input, as long as it is large enough with respect to the gap $\gap$. 
\begin{proposition}\label{prop:equiv}
For every $\eps>0$, there is a constant $C_0=C_0(\eps)$ such that
for every integers $k,d\ge 0$ and $\sqrt{n}\leq g \leq n$, we have
$$
R^k_{2\eps}(\GHS{d,C_0 g/n}) \leq R^k_\eps(\GHD{n}{g}) \leq R^k_\eps(\GHS{n,2g/n}).
$$
\end{proposition}

\begin{proof}
We begin with the right inequality. The idea is that a $\textsc{GHD}_{n,g}$ protocol can be obtained by applying a given $\GHS{}$ protocol to a suitably transformed input.
  Let $x,y\in\{0,1\}^n$ be two inputs to
  $\GHD{n}{g}$. Define $\tilde{x}=\left( {(-1)^{x_i}}/{\sqrt{n}}
  \right)_{i=1,\ldots,n}$ and $\tilde{y}=\left( {(-1)^{y_i}}/{\sqrt{n}}
  \right)_{i=1,\ldots,n}$. Then $\tilde{x},\tilde{y}\in \Sphere^{n-1}$. Moreover,
  $\tilde{x}\cdot \tilde{y} = 1-2\Delta(x,y)/n$.
  Therefore, if $\Delta(x,y)\geq n/2+g$ then $\tilde{x}\cdot \tilde{y}
  \leq -2g/n$, and if $\Delta(x,y)\leq n/2-g$ then $\tilde{x}\cdot \tilde{y} \geq 2g/n$.
  This proves $R^k_\eps(\GHD{n}{g}) \leq R^k_\eps(\GHS{n,2g/n})$.

 For the left inequality, let $x$ and $y$ be two unit vectors (in
  any dimension) such that $|x\cdot y| \geq \gap$, where $\gap = C_0g/n$.
  Note that since $g \geq \sqrt{n}$,
  we have $n = \Omega(\gamma^{-2})$.  Using shared randomness, Alice
  and Bob pick a sequence of vectors $w_1,\ldots,w_n$, each
  independently and uniformly drawn from the unit sphere.  Define two $n$-bit
  strings $\tilde{x} = \left(\sgn(x\cdot w_i) \right)_{i=1,\ldots,n}$ and $\tilde{y} = \left( \sgn(y\cdot w_i) \right)_{i=1,\ldots,n}$.
Let $\alpha = \cos^{-1}(x\cdot y)$ be the angle between $x$ and
$y$. Then a simple argument (used, e.g., by Goemans and Williamson~\cite{Goemans-Williamson})
shows that the probability that a random unit vector $w$ is such that
$\sgn(x\cdot w)\neq \sgn(y\cdot w)$ is exactly
$\alpha/\pi$. This means that for each $i$, the bits $\tilde{x}_i$ and $\tilde{y}_i$ differ with
probability $\frac{1}{\pi}\cos^{-1}(x\cdot y)$, independently of the other bits of $\tilde{x}$ and~$\tilde{y}$.
The first few terms in the Taylor series expansion of $\cos^{-1}$ are $\cos^{-1}(z) =
\frac{\pi}{2} - z - \frac{z^3}{6} + O(z^5)$.
Hence, for each $i$, $\Pr_{w_i}\left(\tilde{x}_i\neq \tilde{y}_i\right) = 1/2-\Theta(x\cdot y)$,
and these events are independent for different~$i$.
Choosing $C_0$ sufficiently large, with probability at least $1-\eps$, the Hamming distance
between $\tilde{x}$ and $\tilde{y}$ is at most $n/2-g$ if $x\cdot y\geq\gap$,
and it is at least $n/2+g$ if $x\cdot y \leq -\gap$.
\end{proof}

\subsection{Concentration of measure}\label{sec:measconc}

It is well known that the Haar measure $\nu$ on a high-dimensional sphere is tightly concentrated around the equator --- around \emph{any} equator, which makes it a fairly counterintuitive phenomenon. The original phrasing of this phenomenon, usually attributed to P. L\'evy~\cite{Levy51}, goes by showing that among all subsets of the sphere, the one with the smallest ``boundary" is the spherical cap $S_\gap^x = \{y\in\Sphere^{n-1}: x\cdot y \geq \gap\}$. The following standard volume estimate will prove useful (see, e.g.,~\cite{BallNotes}, Lemma~2.2).

\begin{fact}\label{fact:cap} Let $x\in\Sphere^{n-1}$ and $\gap>0$. Then
$\displaystyle\nu(S_\gap^x) \leq e^{-\gap^2 n/2}.$
\end{fact}

Given a measurable set $A$, define its \emph{$t$-boundary} $A_t := \{x\in\Sphere^{n-1}: d(x,A)\leq t\}$, for any $t>0$. At the core of our results will be the standard fact that, for any not-too-small set $A$, the set $A_t$ contains almost all the sphere, even for moderately small values of $t$.

\begin{fact}[Concentration of measure on the sphere]\label{fact:gaussianmeasure}
For any measurable $A\subseteq\Sphere^{n-1}$ and any $t>0$,
\begin{align}\label{eq:sphereconc}
\Pr(x\in A) \Pr(x\notin A_t) \leq  4\,e^{-  t^2 n/4},
\end{align}
\noindent where the probabilities are taken according to the Haar measure on the sphere.
\end{fact}

\begin{proof} The usual measure concentration inequality for the sphere (Theorem~14.1.1 in~\cite{Matousekbook}) says that for any set $B\subseteq\Sphere^{n-1}$ of measure at least $1/2$ and any $t'>0$,
\begin{align*}
\Pr(x\notin B_{t'}) \leq 2\,e^{-(t')^2n/2}. 
\end{align*}
This suffices to prove the fact if $\Pr(x\in A) \geq 1/2$.  Assume for the rest of the proof that $\Pr(x\in A) < 1/2$.
  Let $t_0$ be such that $A_{t_0}$ has measure $1/2$; such a $t_0$
  exists by continuity. Applying measure concentration to $B=A_{t_0}$
  gives
\begin{align}\label{eq:notin}
\Pr(x\notin A_{t'+t_0})& \leq  2\,e^{-(t')^2n/2},
\end{align}
for all $t'>0$, while applying it to $B=\overline{A_{t_0}}$ yields
\begin{align}\label{eq:t0}
\Pr(x\in A_{t_0-t''})\leq\Pr(x\not\in B_{t''}) \leq 2\,e^{-(t'')^2n/2}
\end{align}
for all $t'' \leq t_0$, since $A_{t_0 - t''}$ is included in the complement of $(\overline{A_{t_0}})_{t''}$.  Taking $t''=t_0$ gives us $\Pr(x \in A)\leq
2\,e^{-t_0^2n/2}$. If $t\leq t_0$ then this suffices to prove the inequality. Otherwise, set $t':= t - t_0$ in~\eqref{eq:notin} and $t'':=t_0$ in (\ref{eq:t0}) and multiply the two inequalities to obtain the required
bound, by using that $t_0^2+(t-t_0)^2 \geq t^2/2$ (which holds
since $2t_0^2+t^2/2-2t\, t_0=(\sqrt{2}t_0-t/\sqrt{2})^2\geq 0$).
\end{proof}

\paragraph{Why the sphere?} In Section~\ref{sec:combinatorial} we give a proof of a slightly weaker lower bound than the one in our main result by using measure concentration facts on the Hamming cube only. We present those useful facts now, together with a brief discussion of the differences, in terms of concentration of measure phenomenon, between the Haar measure on the sphere and the uniform distribution over the hypercube.  These differences point to the reasons why the proof of Section~\ref{sec:combinatorial} gives an inferior bound.

Similarly to our definition of a spherical cap, we can define a ``Hamming cap'' $T^x_c$ on the Hamming cube as $T^x_c = \{y\in\{0,1\}^n : \Delta(x,y)\leq n/2 - c\sqrt{n}\}$. The analogue of Fact~\ref{fact:cap} is then given by the usual Chernoff bound:

\begin{fact}\label{fact:hammingcap} \label{lem:near-half-ball}
  For all $c > 0$, we have $2^{-n}|T^x_c| \le e^{-2c^2}$.
\end{fact}

A result similar to L\'evy's, attributed to Harper~\cite{Harper66}, states that among all subsets (of the sphere) of a given size, the cap is the one with the smallest boundary. Following a similar proof as for Fact~\ref{fact:gaussianmeasure}, one can get the following statement for the Hamming cube (see e.g. Corollary~4.4 in~\cite{barvinoknotes}):

\begin{fact}[Concentration of measure on the Hamming cube]\label{fact:hammingmeasure}
Let $A\subseteq\{0,1\}^n$ be any set, and define $A_c = \{x\in\{0,1\}^n : \exists y\in A,\, \Delta(x,y)\leq c\sqrt{n}\}$. Then
\begin{align}\label{eq:hammingconc}
\Pr(x\in A) \Pr(x\notin A_c) \leq \,e^{- c^2},
\end{align}
where the probabilities are taken according to the uniform distribution on the Hamming cube.
\end{fact}

To compare these two statements, embed the Hamming cube in the sphere
by mapping $x\in\{0,1\}^n$ to the vector $v_x = \frac{1}{\sqrt{n}}(
(-1)^{x_i})_{i\in [n]}$. Two strings of Hamming distance $c\sqrt{n}$
are mapped to vectors with Euclidean distance $\sqrt{2c}/n^{1/4}$, so
that inequality~\eqref{eq:hammingconc} is much weaker than
inequality~\eqref{eq:sphereconc}.  In particular we see that, while on
the sphere most points are at distance roughly $1/\sqrt{n}$ from any
set of measure half, if we are restricted to the Hamming cube then
very few points are at a corresponding Hamming distance of $1$ from,
say, the set of all strings with fewer than $n/2$ $1$s, which has measure
roughly $1/2$ in the cube. This difference is crucial: it indicates that
the $n$-dimensional cube is too rough an approximation of the
$n$-dimensional sphere for our purposes, perhaps explaining why our
combinatorial bound in Section~\ref{sec:combinatorial} yields a
somewhat weaker dependence on the number of rounds.

\section{Main result}\label{sec:main}

Our main result is the following.

\begin{theorem}\label{thm:main}
  Let $0\leq\eps\leq 1/50$. There exist constants $C$, $C'$ depending
  only on $\eps$ such that the following holds for any $\gamma>0$ and
  any integers $n \geq \eps^2/(4\gap^2)$ and $k \leq C'/(\gap\ln
  (1/\gap))$: if $P$ is a randomized $\eps$-error $k$-round
  communication protocol for $\GHS{n,\gap}$ then some message has
  length at least $\frac{C}{k^2\ln k}\cdot \frac{1}{\gap^2}$ bits.
  \end{theorem}

Using Proposition~\ref{prop:equiv} we get a lower bound for the
Hamming cube version GHD $=\GHD{n}{\sqrt{n}}$:

\begin{corollary}\label{cor:hamming}
Any $\eps$-error $k$-round randomized protocol for GHD communicates $\Omega(n/(k^2\ln k))$ bits.
\end{corollary}

This follows from Theorem~\ref{thm:main} when $k = o(\sqrt{n}/\log n)$. If $k$ is larger, then the bound stated in the Corollary is in fact weaker than the general $\Omega(\sqrt{n})$ lower bound which we sketched in the introduction.

\subsection{Proof outline}\label{sec:proof}

We now turn to the proof of Theorem~\ref{thm:main}. Let $\eps$, $\gap$ and $n$ be as in the statement of the theorem. Since lowering $n$ only makes the $\GHS{n,\gap}$ problem easier, for the rest of this section we assume that $n:=\eps^2/(4\gap^2)$ is fixed, and for simplicity of notation we write $\GHS{\gap}$ for $\GHS{n,\gap}$.

\paragraph{Measurability.}
Before proceeding with the proof, we first need to handle a small technicality
arising from the continuous nature of the input space: namely,
that the distributional protocol might make
decisions based on subsets of the input space that are not measurable.
To make sure that this does not happen, set $\delta = \gap/4$ and consider players Alice and Bob
who first round their inputs to the closest vector in a fixed $\delta$-net, and then proceed
with an $\eps$-error protocol for $\GHS{\gap/2}$. Since the rounding changes $x\cdot y$ by at most
$\gamma/2$, provided Alice and Bob are given valid inputs to $\GHS{\gap}$ they will succeed with probability $1-\eps$.
Hence any randomized $\eps$-error protocol for $\GHS{\gap/2}$ can be transformed
into a randomized $\eps$-error protocol for $\GHS{\gap}$ with the same communication, but which initially rounds its inputs
to a discrete set. We prove
a lower bound on the latter type of protocol. This will ensure that all sets encountered in the proof are measurable.

\paragraph{Distributional complexity.}
By Yao's principle it suffices to
lower-bound the \emph{distributional complexity}, i.e., to analyze \emph{deterministic}
protocols that are correct with probability $1-\eps$ under some input distribution.
As our input distribution for $\GHS{\gap}$ we take the distribution
that is uniform over the inputs satisfying
the promise $|x\cdot y|\geq \gap$.

Given our choice of $n$, Claim~\ref{claim:inner} below guarantees that
 the $\nu\times\nu$-measure of non-promise inputs is at most $\eps$.
 Hence it will  suffice to lower-bound the distributional complexity of protocols
 making error at most $2\,\eps$ under the distribution $\nu\times\nu$.
 We define an \emph{$\eps$-protocol} to be a deterministic
 communication protocol for $\GHS{n,\gap}$ whose
 error under the distribution $\nu\times\nu$ is at most
 $\eps$, where we say that the protocol makes an error if $P(x,y)\neq \sgn(x,y)$.

We prove a lower bound on the maximum length of a message sent by any $\eps$-protocol, via round elimination.  The main reduction step is given by the following technical lemma:

\begin{lemma}[Round Elimination on the sphere]\label{lem:roundelim}
Let $0\leq \eps\leq 1/25$, $\gap>0$, $n=\eps^2/(4\gap^2)$, and $1\leq \kappa \leq k$. Assume there is a $\kappa$-round $\eps$-protocol $P$ such that the first message has length
bounded as $c_1 \leq C_1 \frac{n}{k^2\ln k} -7\ln(2k)$ where $C_1$ is a universal constant. Then there is a $(\kappa-1)$-round
$\eps'$-protocol $Q$ (obtained by eliminating the first
message of $P$), where
$$\eps' \leq \left(1+\frac{1}{k}\right)\,\eps + \frac{1}{16k}.$$
\end{lemma}

Before proving this lemma in Section~\ref{sec:roundelim}, we show how it implies Theorem~\ref{thm:main}.

\begin{proof}[Proof of Theorem~\ref{thm:main}]
We will show that in any $k$-round $(2\,\eps)$-protocol, there is a message sent of length at
least $C_1 n/( k^2\ln k) - 7\ln(2k)$. The discussion in the ``Distributional complexity" paragraph above shows
this suffices to prove the theorem, by setting $C = C_1\eps^2/8$,
and choosing $C'$ small enough so that the bound on $k$ in the
statement of the theorem implies that $7\ln(2k) < C_1 n/( 2 k^2\ln
k)$. 

Let $P$ be a $k$-round $(2\,\eps)$-protocol, and assume for contradiction
that each round of communication uses at most $C_1 n/( k^2\ln k) -
7\ln(2k)$ bits. The recurrence $\eps_\kappa =
(1+1/k)\eps_{\kappa-1}+1/(16k)$, $\eps_0=2\,\eps$, is easily solved to
$\eps_\kappa = (1+1/k)^\kappa(2\,\eps+1/16)-1/16$, so that applying
Lemma~\ref{lem:roundelim} $k$ times leads to a \emph{0-round} protocol
for $\GHS{\gap}$ that errs with probability at most $\eps' \leq
e\,(2\,\eps+1/16)-1/16\leq 1/4$ over the input distribution $\nu\times\nu$.
 We have reached a contradiction:
such a protocol needs communication and hence cannot be 0-round. Hence
$P$ must send a message of length at least $C_1 n/( k^2\ln k) -
7\ln(2k)$, which is what needed to be shown.
\end{proof}

\subsection{The main reduction step}\label{sec:roundelim}

\begin{proof}[Proof of Lemma~\ref{lem:roundelim}]
Let $P(x,y)$ denote the output of the protocol on input $x,y$.
Define $x \in \Sphere^{n-1}$ to be \emph{$\delta$-good} if
$\Pr_{\nu\times\nu}(P(x,y) \text{ errs } | x) \leq \delta \eps$.  By
Markov's inequality, at least a $(1-1/\delta)$-fraction of $x$ (distributed
according to $\nu$) are good.  For a given message $m$, let $A_m$ be the set
of all good $x$ on which Alice sends $m$ as her first message.
The sets $A_m$, over all messages $m\in\{0,1\}^{c_1}$,
form a partition of the set of good $x$.
Define $m_1:=\text{argmax}_m  \nu(A_m)$ and let $A := A_{m_1}$.
Setting $\delta = 1+1/k$, we have $\nu(A)\geq \left(1-\frac{1}{\delta}\right)2^{-c_1} \geq e^{-c_1-\ln(k+1)}$.

We now define protocol $Q$. Alice receives an input $\tilde{x}$, Bob receives $\tilde{y}$,
both distributed according to $\nu$.
Alice computes the point $x\in A$ that is closest to $\tilde{x}$, and Bob sets $y :=\tilde{y}$.
They run protocol $P(x,y)$ without Alice sending the first message,
so Bob starts and proceeds as if he received the fixed message $m_1$ from Alice.

To prove the lemma, it suffices to bound the error probability $\eps'$ of $Q$ with input $\tilde{x},\tilde{y}$ distributed according to $\nu\times\nu$.
Define $d_1=2\sqrt{\frac{c_1+6\ln(2k)+2}{n}}$.
We consider the following bad events:
\begin{itemize}
\item $\text{BAD}_1: d(\tilde{x},A) > d_1 $
\item $\text{BAD}_2: P(x,y)\neq \sgn(x\cdot y)$
\item $\text{BAD}_3: d(\tilde{x},A) \leq d_1$ but $\sgn(x\cdot y) \neq \sgn(\tilde{x} \cdot
  \tilde{y})$.
\end{itemize}
If none of those events occurs, then protocol $P$ outputs the correct
answer. We bound each of them separately, and will conclude by
upper bounding $\eps'$ with a union bound.

The first bad event can be easily bounded using the measure
concentration inequality from Fact~\ref{fact:gaussianmeasure}.  Since
$\tilde{x}$ is uniformly distributed in $\Sphere^{n-1}$ and
$\Pr(A)\geq e^{-c_1-\ln(k+1)}$, we get
$$
\Pr(\text{BAD}_1)\leq  4\,e^{-d_1^2 n/4 + c_1+\ln(k+1)}\leq 4\,e^{-5\ln(2k)-2}\leq \frac{1}{32k}.
$$

The second bad event has probability bounded by $(1+1/k)\,\eps$ by
the goodness of $x$. Now consider event $\text{BAD}_3$. Without
loss of generality, we may assume that $\tilde{x}\cdot \tilde{y} =
\tilde{x}\cdot y > 0 $ but $x\cdot y <0 $ (the other case is treated
symmetrically). In order to bound $\text{BAD}_3$, we will use the two following claims. The first shows that the probability that $\tilde{x}\cdot y$  is close to $0$ for a random $\tilde{x}$ and $y$ is small. The second uses measure concentration to show that, if $\tilde{x}\cdot y$ is not too close to $0$, then moving $\tilde{x}$ to the nearby $x$ is unlikely to change the sign of the inner product.

\begin{claim}\label{claim:inner}
  Let $x,y$ be distributed according to $\nu$. For any real
  $\alpha\geq 0$,
 $$\Pr( 0\leq x\cdot y\leq \alpha) \leq  \alpha\sqrt{n}$$
 \end{claim}

 \begin{proof} Letting $\omega_n$ be the volume of the $n$-dimensional Euclidean unit ball, we can write (see e.g.,~\cite{BGKKLS98}, Lemma~5.1)
 \begin{align*}
 \Pr(  0\leq x\cdot y\leq \alpha) &= \frac{(n-1)\,\omega_{n-1}}{n\,\omega_n}\int_{0}^\alpha (1-t^2)^{\frac{n-3}{2}}\text{dt}\\
 & \leq \alpha\sqrt{n}
 \end{align*}
 where we used $\frac{\omega_{n-1}}{\omega_n} < \sqrt{\frac{n+1}{2\pi}} < \sqrt{n}$.
\end{proof}

\begin{claim}\label{claim:ip}
Let $x,\tilde{x}$ be two fixed unit vectors at distance
$\|x-\tilde{x}\| = d\in[0,d_1]$, and $0 < \alpha \leq 1/(4\sqrt{n})$.
Let $y$ be taken according to $\nu$. Then
$$\Pr(\tilde{x}\cdot y \geq \alpha \wedge x\cdot y< 0) \leq
\,e^{-\alpha^2 n/(8d_1^2)}.$$
\end{claim}

\begin{proof}
Note that $x\cdot\tilde{x}= 1-\|x-\tilde{x}\|^2/2 = 1-d^2/2$.
Since the statement of the lemma is rotationally-invariant, we may assume without loss of generality that
\begin{align*}
\tilde{x} & =  (1,0,0\ldots,0),\\
x         & =  (1-d^2/2,-\sqrt{d^2-d^4/4},0,\ldots,0),\\
y         & =  (y_1,y_2,y_3,\ldots,y_n).
\end{align*}
Therefore, $y_1\geq \alpha$ when $\tilde{x}\cdot y\geq\alpha$.
Note that
$$
x\cdot y = x_1 y_1 + x_2 y_2\geq(1-d^2/2)\alpha - \sqrt{d^2-d^4/4}\,y_2.
$$
Hence the event $\tilde{x}\cdot y \geq \alpha \wedge x\cdot y< 0$ implies
\begin{align*}
y_2 &>\frac{(1-d^2/2)\,\alpha}{\sqrt{d^2-d^4/4}}\\
& \geq \frac{\alpha}{2d}
\end{align*}
where we used the fact that $d\leq d_1\leq 1$, given our assumption on $c_1$.
By Fact~\ref{fact:cap}, the probability that, when $y$ is sampled from $\nu$, $y_2$ is larger than $\alpha/(2d)$ is at most $ e^{-\alpha^2 n/(8d^2)}$. Hence the probability that both $\tilde{x}\cdot y \geq \alpha$ and $x\cdot y <0$ happen is at most as much.
\end{proof}

Setting $\alpha = 1/(128k\sqrt{n})$, by Claim~\ref{claim:inner} we
find that the probability that $0\leq \tilde{x}\cdot y \leq \alpha$ is
at most $1/(128k)$.  Furthermore, the probability that $\tilde{x}\cdot
y \geq \alpha$ and $x\cdot y < 0$ is at most $\exp\left( -
\frac{n}{2^{19} k^2 (c_1+6\ln(2k)+2)}\right)$ by Claim~\ref{claim:ip}.
This bound is less than $1/(128k)$ given our assumption on $c_1$,
provided $C_1$ is a small enough constant.  Putting both bounds
together, we see that 
$$\Pr(\tilde{x}\cdot y \geq 0 \wedge x\cdot y <
0) < 1/(64k).$$
 The event that $\tilde{x}\cdot y <0$ but $x \cdot y
\geq 0$ is bounded by $1/(64k)$ in a similar manner.
Hence, $\Pr(\text{BAD}_3) < 1/(32k)$.  Taking the union bound over all three
bad events concludes the proof of the lemma.
\end{proof}

\section{A simple combinatorial proof}\label{sec:combinatorial}

In this section we present a combinatorial proof of the following:

\begin{theorem}\label{thm:hd-ghd}
  Let $0\leq\eps\leq 1/50$. There exists a constant $C''$ depending on $\eps$ only, such that the following holds for any $g \leq C''\sqrt{n}$ and $k \leq n^{1/4}/(1024\log n)$: if $P$ is a
  randomized $\eps$-error $k$-round communication protocol for $\GHD{n}{g}$
  then some message has length at least $\frac{n}{(512k)^4\log^2 k}$ bits.
\end{theorem}

Even though this is a weaker result than
Theorem~\ref{thm:main}, its proof is
simpler and is based on concentration of measure in the Hamming cube
rather than on the sphere (we refer to Section~\ref{sec:measconc} for
a high-level comparison of the two proofs). Interestingly, the dependence
on the number of rounds that we obtain is quadratically worse than that
of the proof using concentration on the sphere. We do not know if this
can be improved using the same technique.

We proceed as in Section~\ref{sec:proof}, observing that it suffices
to lower-bound the distributional complexity of $\GHD{n}{g}$ under a
distribution uniform over the inputs satisfying the promise
$|\Delta(x,y)-n/2|\geq g$. In fact, as we did before, by
taking $C''$ small enough we can guarantee that the number of non-promise inputs is at
most $\eps\, 2^n$. Hence it will suffice to lower-bound the
distributional complexity of protocols making error at most $2\,\eps$
under the uniform input distribution.  We define an \emph{$\eps$-protocol} to
be a deterministic communication protocol for GHD whose distributional
error under the uniform distribution is at most $\eps$.
 The following is the analogue of Lemma~\ref{lem:roundelim},
from which the proof of Theorem~\ref{thm:hd-ghd} follows as in
Section~\ref{sec:proof}.

\begin{lemma}[Round Elimination on the Hamming cube]
Let $0<\eps\leq 1/25$ and $\kappa,k$ be two integers such that $k\geq 128$ and $1\leq\kappa\leq k \leq n^{1/4}/(1024\log n)$. Assume that there is a $\kappa$-round $\eps$-protocol $P$ such that the first message has length
bounded by $c_1\le  n/((512k)^4\log^2 k)$.
    Then there exists a $(\kappa-1)$-round
$\eps'$-protocol $Q$ (obtained by eliminating the first
message of $P$) where
$$\eps' \leq \left(1+\frac{1}{k}\right)\,\eps + \frac{1}{16k}.$$
\end{lemma}

\begin{proof}
Define $x \in \{0,1\}^n$ to be \emph{good} if
$\Pr(P(x,y) \text{ errs } | x) \leq (1+1/k) \eps$.  By
Markov's inequality, at least a $1/(k+1)$-fraction of $x\in\{0,1\}^n$ are good.  For a given message $m$, let $A_m
\deq \{\text{good }x: \text{ Alice sends $m$ given }x\}$.  The sets $A_m$, over all messages $m\in\{0,1\}^{c_1}$,
together form a partition of the set of good $x$.
Define $m_1:=\text{argmax}_m  |A_m|$, and let $A:=A_{m_1}$.
By the pigeonhole principle, we have $|A|\geq  \frac{1}{k+1} 2^{n-c_1}$.

We now define protocol $Q$. Alice receives an input $\tilde{x}$, Bob receives $\tilde{y}$, uniformly distributed.
Alice computes the string $x\in A$ that is closest to $\tilde{x}$ in Hamming distance, and Bob sets $y :=\tilde{y}$.
They run protocol $P(x,y)$ without Alice sending the first message,
so Bob starts and proceeds as if he received the fixed message $m_1$ from Alice.

To prove the lemma, it suffices to bound the error probability $\eps'$ of $Q$ under the uniform distribution.
Define $d_1=9\sqrt{n}/((1024 k)^2\log k)$. As in the proof of Lemma~\ref{lem:roundelim}, we consider the following bad events:
\begin{itemize}
\item $\text{BAD}_1: \Delta(x,\tilde{x}) > d_1\sqrt{n} $
\item $\text{BAD}_2: P(x,y)\neq \text{GHD}(x,y)$
\item $\text{BAD}_3: \Delta(x,\tilde{x}) \leq d_1\sqrt{n}$ but GHD$(\tilde{x},y)\neq \text{GHD}(x,y)$
\end{itemize}
If none of those events occurs, then protocol $P$ outputs the correct
answer. We bound each of them separately, and will conclude by a union bound.

The first bad event is easily bounded using Fact~\ref{fact:hammingmeasure}, which implies that
$$
\Pr(\tilde{x}\notin A_{d_1}) \leq e^{-81n/((1024k)^4\log^2 k)} 2^{c_1+\log(k+1)} \leq \frac{2}{k^2} \leq \frac{1}{32k}
$$
given our assumptions on $c_1$ and $k$. The second bad event is bounded by $(1+1/k)\,\eps$, by definition of the set~$A$.

We now turn to $\text{BAD}_3$. The event that GHD$(\tilde{x},y)\ne $ GHD$(x,y)$ only depends on the
relative distances between $x$, $\tilde{x}$, and $y$, so we may apply
a shift to assume that $x = (0,\ldots,0)$.  Without loss of
generality, we assume that $\Delta(\tilde{x},y)>n/2$ and
$|y|<n/2$ (the error bound when $\Delta(\tilde{x},y)<n/2$ and
$|y|>n/2$ is proved in a symmetric manner). Note that, since $y$ is
uniformly random (subject to $|y|<n/2$), with probability at least $1-1/(128k)$, we have
$|y|\leq n/2-\sqrt{n}/(128k)$. Hence we may assume that this holds with an
additive loss of at most $1/(128k)$ in the error. Now
  \begin{align*}
    \dist(\tilde{x},y) > n/2
    &\Longleftrightarrow |\tilde{x}| + |y| - 2|\tilde{x} \cap y| > n/2 \\
    &\Longleftrightarrow |\tilde{x} \cap y| < \frac{|\tilde{x}|+|y|-n/2}{2}.
  \end{align*}
  It is clear that the worst case in this statement is for
  $|y|=n/2-\sqrt{n}/(128k)$ and $|\tilde{x}| = \Delta(x,\tilde{x}) = d_1\sqrt{n}$. By symmetry,
  the probability that this event happens is the same as if we fix any
  $y$ of the correct weight, and $\tilde{x}$ is a random string of
  weight $d_1\sqrt{n}$. Since the expected intersection size is
  $|\tilde{x}|/2 -d_1/(128k)$, by Hoeffding's inequality (see e.g. the bound on the tail of the hypergeometric distribution given in~\cite{Chvatal79}), for $a=\sqrt{n}/(256k) - d_1/(128k)$

  \begin{align*}
    \Pr\left(|\tilde{x} \cap y| \leq \frac{|\tilde{x}|+|y|-n/2}{2}\right) &=
    \Pr\left(|\tilde{x} \cap y| \le \mathbb{E}[|\tilde{x}\cap y|] - a\right) \\
    &\le  e^{-2a^2/(d_1\sqrt{n})}.
  \end{align*}
  Given our choice of $d_1$ we have $a \geq 3\sqrt{n}/(4\cdot 256k)$, and hence the upper bound is at most
  $1/k^2 \leq 1/(128k)$, given our assumption on $k$. Applying the union bound over all bad events then yields
  the lemma.
\end{proof}

\subsubsection*{Acknowledgments}
We thank Ishay Haviv for discussions during the early stages of this work.

\bibliography{ghd}
\end{document}